\documentclass[letterpaper, 10 pt,journal]{IEEEtran}  % Comment this line out
\usepackage{graphicx}
\usepackage{amsmath,amssymb}   
\usepackage{tikz}
\usepackage{cite}
\usepackage{subfigure}

\usepackage{amsthm}
\theoremstyle{remark}
\newtheorem{thm}{Theorem}
\newtheorem{defn}[thm]{Definition}
\newtheorem{assu}[thm]{Assumption}

\newtheorem{lem}[thm]{Lemma}
\newtheorem{lrop}[thm]{Proposition}
\newtheorem{coro}[thm]{Corollary}
\theoremstyle{remark}
\newtheorem{rmk}[thm]{Remark}

\usepackage{mathtools}

\def\TM{\textcolor{black}}

\title{\LARGE \bf
Dissipativity verification with guarantees for polynomial systems\\ from noisy input-state data
}

\author{Tim Martin and Frank Allg{\"o}wer*% <-this % stops a space
	\thanks{*T. Martin and F. Allg{\"o}wer are with the Institute for Systems Theory and Automatic Control, University of Stuttgart. This work was funded by Deutsche Forschungsgemeinschaft (DFG, German Research Foundation) under Germany's Excellence Strategy - EXC 2075 - 390740016. For correspondence, mailto: tim.martin@ist.uni-stuttgart.de.}% <-this % stops a space
}

\begin{document}

\IEEEoverridecommandlockouts

\IEEEpubid{\begin{minipage}{\textwidth}\ \\[12pt] \copyright 2020 IEEE. This version has been accepted for publication in EEE
		Control Systems Letters, 2020. Personal use of this material is permitted. Permissionfrom EUCA must be obtained for all other uses, in any current or future media, including reprinting/republishing this material for advertising or promotionalpurposes, creating new collective works, for resale or redistribution to servers or lists, or reuse of any copyrighted component of this work in other works.\end{minipage}}

\maketitle
\pagestyle{empty}

%%%%%%%%%%%%%%%%%%%%%%%%%%%%%%%%%%%%%%%%%%%%%%%%%%%%%%%%%%%%%%%%%%%%%%%%%%%%%%%%
\begin{abstract}
	
In this paper, we investigate the verification of dissipativity properties for polynomial systems without an \TM{explicitly identified} model but directly from noise-corrupted measurements. Contrary to most data-driven approaches for nonlinear systems, we determine dissipativity properties over \TM{all finite} time horizon\TM{s} using \TM{noisy} input-state data. To this end, we propose two noise characterizations to deduce \TM{two} data-based set-membership representations of the ground-truth system. Each representation then serves as a framework to derive computationally tractable conditions to verify dissipativity properties with rigorous guarantees from noise-corrupted data using \TM{sum of squares (SOS)} optimization.

\end{abstract}

%%%%%%%%%%%%%%%%%%%%%%%%%%%%%%%%%%%%%%%%%%%%%%%%%%%%%%%%%%%%%%%%%%%%%%%%%%%%%%%%
\section{Introduction}\label{Intrduction}

\IEEEPARstart{T}{he} standard approach to obtain a controller for nonlinear systems requires to retrieve a sufficiently precise model and the application of nonlinear controller design techniques~\cite{Khalil}. However, the identification of nonlinear systems is in general time consuming and \TM{requires often expert knowledge}. Hence, the interest on data-driven controller design techniques, where the controller is deduced without \TM{identifying} a model but directly from measured data of the system, has risen recently. %An overview of such approaches can be found in \cite{DataSurvey}.
\\\indent
One well-elaborated theory for the controller design of nonlinear systems are dissipativity properties \cite{DissiWillems} \TM{which give rise to stabilizing control laws as the small gain theorem \cite{Khalil} (Theorem 5.6)}. % or the feedback theorem for passive systems \cite{Khalil} (Theorem 6.1)}. 
Since these system properties give insight to the system and facilitate a controller design without knowledge of the system, the verification of these properties from measured trajectories can be leveraged to a data-driven controller design with stability and performance guarantees.\\\indent 
For linear time-invariant (LTI) systems, \cite{OneShot} determines dissipativity properties over a \TM{data-depended} finite time horizon from a noise-free single input-output trajectory. By exploiting the set-membership representation of an unknown LTI system by noisy input-state samples from \cite{Groningen}, \cite{AnneDissi} provides guaranteed dissipativity properties over \TM{all finite} time horizon\TM{s}\TM{\ as defined in \cite{DissiWillems} and required for, e.g., the small gain theorem.} For nonlinear systems, \cite{MontenbruckLipschitz} is tailored to estimate certain dissipativity properties over a \TM{data-depended} finite time horizon, as the $\mathcal{L}_2$-gain or conic relations \cite{Zames}, from a large number of input-output trajectories based on the Lipschitz constant of the system operator. To reduce the amount of required data, \cite{IterativeNLM} proposes sequential experiments to improve iteratively the accuracy of a non-parametric data-based Lipschitz approximation of the system operator. Nevertheless, the amount of data might be still too large for a real application as also indicated by bounds on the sampling complexity from \cite{Sharf}. \\\indent
\TM{For that reason, we establish in this paper a data-based framework, which is based on a set-membership approach for polynomial systems and is more data-efficient than \cite{MontenbruckLipschitz} and \cite{IterativeNLM}, to determine dissipativity properties over \TM{all finite} time horizon\TM{s}.} Contrary to \cite{DePersisSOS}, we consider polynomial systems in discrete time and measurements in presence of noise. By characterizing this noise by two distinct descriptions, we propose two data-based set-membership representations of the ground-truth system which constitute two frameworks to deduce computationally tractable conditions for verifying dissipativity properties using sum of squares (SOS) optimization. The first noise description bounds the noise signal in each time step which is commonly assumed, e.g., in set-membership identification \cite{Milanese}. This characterization yields for the verification of dissipativity properties with polynomial supply rates a\TM{n} SOS optimization problem which can be solved by semi-definite programming using standard SOS techniques \cite{SOSTutorial}. Since the complexity of this SOS optimization problem increases for additional samples, the second ansatz characterizes the noise by a single cumulative property. %as the energy of the noise over the measured time horizon. 
This approach was first introduced in \cite{Groningen} and yields a feasibility condition of a linear matrix inequality (LMI) to verify $(Q,S,R)$-dissipativity.  %\\\indent 
%The paper is organized as follows. First, we introduce some notations for SOS optimization and specify the problem setup for dissipativity verification of polynomial systems. Subsequent, we propose in Section~\ref{SecGeneralSOS} and Section~\ref{SecNonCons} each one noise description to deduce a data-based set-membership representation of the ground-truth system which yields a computationally tractable condition for verifying dissipativity. In Section~\ref{SecCompare}, we compare both approaches and apply them on two numerical examples in Section~\ref{SecEx}.
\IEEEpubidadjcol

\section{Preliminaries}\label{PreSet}

In this section, we introduce the notion of SOS polynomials and matrices and formulate the problem of verifying dissipativity properties for unidentified polynomial systems from noise-corrupted input-state data.

\subsection{SOS optimization}\label{SecSOS}

For a vectorial index $\alpha=\begin{bmatrix}
\alpha_1 & \cdots & \alpha_n\end{bmatrix}^T\in\mathbb{N}_0^n$ and a vector $x=\begin{bmatrix}
x_1 & \cdots & x_n\end{bmatrix}^T\in\mathbb{R}^n$, we write $|\alpha|=\TM{\alpha}_1+\cdots+ \TM{\alpha}_n$, the monomial $x^\alpha=x_1^{\alpha_1}\cdots x_n^{\alpha_n}$, and $\mathbb{R}[x]$ for the set of all polynomials $p$ in $x$, i.e.,
\begin{equation*}
p(x)=\sum_{\alpha\in\mathbb{N}_0^n,|\alpha|\leq d} a_\alpha x^\alpha,
\end{equation*}
with real coefficients $a_\alpha\in\mathbb{R}$. $d\in\mathbb{N}_0$ corresponds to the degree of the polynomial if there is an $a_\alpha\neq0$ with $|\alpha|=d$. Furthermore, we denote $\mathbb{R}[x]^m$ as the set of all $m$-dimensional vectors with entries in $\mathbb{R}[x]$ and $\mathbb{R}[x]^{r\times s}$ as the set of all $r\times s$-matrices with entries in $\mathbb{R}[x]$. The degree of a polynomial matrix is the largest degree of its elements. 
\begin{defn}[SOS matrix]\label{DefSOSMatrix}
	A polynomial matrix $P\in\mathbb{R}[x]^{\TM{r}\times \TM{r}}$ with even degree is called a\TM{n} SOS matrix if there exists a matrix $Q\in\mathbb{R}[x]^{\TM{s}\times \TM{r}}$ such that \TM{$P=Q^TQ$}. Moreover, let the set of all $\TM{r}\times \TM{r}$-SOS matrices be denoted by $\text{SOS}[x]^{\TM{r}\times \TM{r}}$. For $\TM{r}=1$, $P$ is called SOS polynomial.	
\end{defn}

SOS matrices are \TM{computationally attractive} as we can verify whether a polynomial matrix is a\TM{n} SOS matrix by an LMI feasibility condition which deduces from the \TM{following square matricial representation \cite{SOSDecomp}}.   
\begin{lrop}\label{SOSMatrix}
	A polynomial matrix $P\in\mathbb{R}[x]^{\TM{r}\times \TM{r}}$ is a\TM{n} SOS matrix if and only if there exist a real matrix $X\succeq0$ \TM{and a vector $Z\in\mathbb{R}[x]^{\beta}$ containing monomials of $x$} such that
	\begin{equation*}
	\TM{P=\begin{bmatrix}Z\otimes I_{\TM{r}}	\end{bmatrix}^TX\begin{bmatrix}Z\otimes I_{\TM{r}}	\end{bmatrix},}
	\end{equation*}	
	where $I_{\TM{r}}$ denotes the $\TM{r}\times \TM{r}$-identity matrix and $\otimes$ corresponds to the Kronecker product. 
\end{lrop}
\begin{proof}
	\TM{The statement follows from the Gram matrix method \cite{Gmm}. A detailed proof can be found in \cite{SOSDecomp}.}
\end{proof}

In our application of SOS optimization, we are confronted to verify that \TM{a polynomial $p\in\mathbb{R}[x]$ is non-negative} for all $x\in\{x\in\mathbb{R}^n:c_1(x)\geq0,\dots,c_k(x)\geq0\}$ \TM{with $c_i\in\mathbb{R}[x]$}. We can boil down this problem to an LMI feasibility condition using the following SOS relaxation \TM{from \cite{ProofProp}}.
\begin{lrop}[SOS relaxation]\label{SOSRelaxation}
	\TM{A polynomial $p\in\mathbb{R}[x]$ is non-negative for all $x\in\{x\in\mathbb{R}^n:c_1(x)\geq0,\dots,c_k(x)\geq0\}$ with $c_i\in\mathbb{R}[x]$ if there exist SOS polynomials $t_i\in\text{SOS}[x],i=1,\dots,k$ such that $p-\sum_{i=1}^{k}t_ic_i\in\text{SOS}[x]$.}
\end{lrop}
\begin{proof}\TM{
		A proof based on the Positivstellensatz can be found in \cite{ProofProp} (Lemma 2.1).}
\end{proof}
%Since not every positive definite polynomial matrix $P(x)$ is a\TM{n} SOS matrix, the SOS relaxation is in general not a tight description of positive definite polynomial matrices. %However, the relaxation in Proposition~\ref{SOSRelaxation} is indeed asymptotically exact in the sense that the SOS relaxation is tight if the degree of the SOS matrices $T_i(x)$ tends to infinity \cite{SchererLMI}.

\subsection{Problem setup}\label{ProblemSetup}

We consider the nonlinear discrete-time system with polynomial dynamics 
\begin{align}\label{TrueSystem}
x(t+1)=f(x(t),u(t)),\ f\in\mathbb{R}[x,u]^{n}
\end{align}
and state-input constraints $(x,u)\in\TM{\mathbb{P}}$ with \TM{non-empty set}
\begin{equation}\label{Constraints}
\begin{aligned}
\TM{\mathbb{P}}=\{(x,u)\in\mathbb{R}^n\times\mathbb{R}^m: p_i(x,u)\leq 0,\  p_i\in\mathbb{R}[x,u],&\\ i=1,\dots,c\}&.
\end{aligned}
\end{equation}
The goal of this paper is the derivation of computationally tractable conditions to check whether system \eqref{TrueSystem} is dissipative on \eqref{Constraints} without identifying a model but directly from input-state data. Since dissipativity properties are originally defined for continuous-time unconstrained systems in \cite{DissiWillems}, we specify a suitable notion of dissipativity for discrete-time systems under constraints.
\begin{defn}[Dissipativity]\label{DissiDef}
	System~\eqref{TrueSystem} is dissipative on $\TM{\mathbb{P}}\subseteq\mathbb{R}^n\times\mathbb{R}^m$ with respect to the given supply rate $s:\TM{\mathbb{P}}\rightarrow\mathbb{R}$ if there exists a continuous storage function $\lambda:\mathbb{X}\rightarrow\mathbb{R}_{\geq0}$ such that
	\begin{align}\label{dissipativityInqu}
	\lambda(f(x,u))-\lambda(x)\leq s(x,u),\quad \forall (x,u)\in\TM{\mathbb{P}}, 
	\end{align}
	where $\mathbb{X}\subseteq\mathbb{R}^n$ denotes the projection of $\TM{\mathbb{P}}$ on the state-space $\mathbb{R}^n$. Moreover, the system is called $(Q,S,R)$-dissipative if it is dissipative with respect to the supply rate 
	\begin{equation*}
	s(x,u)=\begin{bmatrix}x\\u\end{bmatrix}^T\begin{bmatrix}Q & S\\ S^T & R\end{bmatrix}\begin{bmatrix}x\\u\end{bmatrix}.
	\end{equation*}
\end{defn} 

\vspace{0.2cm}While the verification of dissipativity inequality \eqref{dissipativityInqu} for a (known) polynomial system~\eqref{TrueSystem} and polynomial supply rate using SOS optimization is well-investigated \cite{SOSTutorial}, dissipativity verification of an unidentified polynomial system directly from noisy data as formulated next hasn't been analyzed yet.\\\indent%, to the authors' best knowledge. 
Suppose that an upper bound on the degree of $f$ is known while its coefficients are unidentified. Then the system dynamics~\eqref{TrueSystem} can be \TM{represented by}
\begin{align}\label{Systemdesc}
f(x,u)=Az(x,u)=(I_n\otimes z(x,u)^T)a,
\end{align}
where $z\in\mathbb{R}[x,u]^\ell$ contains \TM{at least all monomials of $f$ according to the known upper bound on the degree of $f$}. $A\in\mathbb{R}^{n\times \ell}$ and \TM{$a=\text{vec}(A^T)\in\mathbb{R}^{n\ell}$, where vec denotes the vectorization of a matrix by stacking its columns,} contain the unknown coefficients. Furthermore, we assume the access to \TM{noisy}
input-state data
\begin{equation}\label{DataSet}
\{(\tilde{x}_i^+,\tilde{x}_i,\tilde{u}_i)_{i=1,\dots,D}\}
\end{equation}
satisfying $\tilde{x}_i^+=f(\tilde{x}_i,\tilde{u}_i)+\tilde{d}_i$. \TM{Note that we measure the state $\tilde{x}_i^+$ from the underlying system, i.e., $\tilde{x}_i^+=x_i^++d_i$ with $x_i^+=f(x_i,u_i)$, the true state $x_i\neq\tilde{x}_i$, and input $u_i\neq\tilde{u}_i$ and where $d_i\in\mathbb{R}^n$ summarizes the uncertainty due to measurement noise. Therefore, $\tilde{d}_i=d_i+f(x_i,u_i)-f(\tilde{x}_i,\tilde{u}_i)$, i.e., the noise vector $\tilde{d}_i$ contains the effect of $d_i$ and \TM{analogously to \cite{Milanese}} the difference when applying the dynamics to the uncertain state $\tilde{x}_i$ and input $\tilde{u}_i$ instead of the true state $x_i$ and input $u_i$.} As clarified in \cite{AnneDissi}, we could also study noise $\tilde{d}_i$ that affects through a matrix $B$ to include addition knowledge on its influence. \\\indent
In the sequel, we characterize the noise $\tilde{d}_i,i=1,\dots,D$ more precisely to derive data-based set-membership representations of the unidentified polynomial system~\eqref{TrueSystem}.% which constitute frameworks to verify whether \eqref{TrueSystem} is dissipative.

\section{Data-driven dissipativity verification for separately bounded noise}\label{SecGeneralSOS}

In this section, we develop a framework for dissipativity verification of polynomial system~\eqref{TrueSystem} from noise-corrupted data \eqref{DataSet} if the noise is bounded explicitly in each time step as specified in the following assumption.
\begin{assu}[Separately bounded noise]\label{Noise1}
	For the measured data \eqref{DataSet}, suppose that for $i=1,\dots,D$
	\begin{equation}\label{BoundedNoise1}
	\tilde{d}_i\in\mathcal{D}_i^{\TM{\text{SB}}}=\{d\in\mathbb{R}^n:\delta_i(d)\leq0,\ \delta_i\in\mathbb{R}[d]\},
	\end{equation} 
	where $\mathcal{D}_i^{\TM{\text{SB}}}$ is bounded.
\end{assu}

The noise characterization in Assumption~\ref{Noise1} \TM{seems to be general} and incorporates, e.g., quadratically bounded noise
\begin{equation}\label{QuadBoundNoise}
\delta_i(d)=\begin{bmatrix}d\\1	\end{bmatrix}^T\begin{bmatrix}\Delta_1&\Delta_2\\ \Delta_2^T & \Delta_3,	\end{bmatrix}\begin{bmatrix}d\\1	\end{bmatrix}, \Delta_1\succ0, \Delta_3\leq0.
\end{equation}
Moreover, Assumption~\ref{Noise1} includes noise with bounded amplitude $\delta_i(d)=d^Td-\epsilon^2$ and noise that exhibits a fixed signal-to-noise-ratio $\delta_i(d)=d^Td-\tilde{\epsilon}^2\tilde{x}_i^T\tilde{x}_i$ which are frequently assumed in system identification \cite{Milanese}. \\\indent
To derive a data-based set-membership representation of the ground-truth system~\eqref{TrueSystem} which is the basis to verify dissipativity properties without \TM{identifying} an explicit model, we next define the set of all systems \TM{parametrized by $a$}
\begin{align}\label{Systemdescribtion}
x(t+1)=\underbrace{(I_n\otimes z(x,u)^T)}_{=:Z(x,u)}a,
\end{align}
with unidentified coefficient vector $a\in\mathbb{R}^{n\TM{\ell}}$ and known vector $z\in\mathbb{R}[x,u]^{\ell}$, which explain the data \eqref{DataSet}. 
\begin{defn}[Feasible system set]\label{DefFSS1}
	The set of all systems \eqref{Systemdescribtion} admissible with the measured data \eqref{DataSet} for separately bounded noise \eqref{BoundedNoise1} is given by the feasible system set $\text{FSS}_{\TM{\text{SB}}}=\{Za\in\mathbb{R}[x,u]^n:a\in\Sigma_{\TM{\text{SB}}}	\}$ with $\Sigma_{\TM{\text{SB}}}=\{a\in\mathbb{R}^{n\ell}:\forall i\in\{1,\dots,D\}\hspace{0.25cm} \exists \tilde{d}_i\TM{\in\mathcal{D}_i^{\TM{\text{SB}}}} \text{\ satisfying\ }\tilde{x}_i^+=Z(\tilde{x}_i,\tilde{u}_i)a+\tilde{d}_i\}$.
\end{defn}

Since the samples \eqref{DataSet} satisfy $\tilde{x}_i^+=f(\tilde{x}_i,\tilde{u}_i)+\tilde{d}_i$ with $\tilde{d}_i\in\mathcal{D}_i^{\TM{\text{SB}}}$ by assumption, 
the ground-truth system is an element of $\text{FSS}_{\TM{\text{SB}}}$, i.e., $f\in\text{FSS}_{\TM{\text{SB}}}$. Thereby, $\text{FSS}_{\TM{\text{SB}}}$ is a set-membership representation of the ground-truth system~\eqref{TrueSystem}. Analogously to \cite{Groningen}, we deduce in the following lemma a data-based description of $\text{FSS}_{\TM{\text{SB}}}$.
%\begin{lem}\label{LemFSS1}
%The set of all coefficients $\Sigma_a$ for which system \eqref{Systemdescribtion} explains the measured data set \eqref{DataSet} for separately bounded noise~\eqref{BoundedNoise1} is equivalent to
%\begin{align}\label{DataBasedParaSet}
%	\{a:\begin{bmatrix}a\\1\end{bmatrix}^TQ_i\begin{bmatrix}a\\1\end{bmatrix}\leq0,i=1,\dots,D\}
%\end{align}
%with the data-dependent matrices
%\begin{align*}
%	Q_i {=}\begin{bmatrix}Z(\tilde{x}_i,\tilde{u}_i)^TZ(\tilde{x}_i,\tilde{u}_i) & -Z(\tilde{x}_i,\tilde{u}_i)^T\tilde{x}_i^+\\-{\tilde{x}_i^{+^T}}Z(\tilde{x}_i,\tilde{u}_i) & {\tilde{x}_i^{+^T}}{\tilde{x}_i^+}-\delta_i^2 \end{bmatrix}.
%\end{align*}
%\end{lem}
\begin{lem}\label{LemFSS1}
	The set of all coefficients $\Sigma_{\TM{\text{SB}}}$ for which system \eqref{Systemdescribtion} explains the measured data set \eqref{DataSet} for separately bounded noise~\eqref{BoundedNoise1} is equivalent to
	\begin{align}\label{DataBasedParaSet}
	\{a\in\mathbb{R}^{n\ell}:\delta_i(\tilde{x}_i^+-Z(\tilde{x}_i,\tilde{u}_i)a)\leq0,i=1,\dots,D\}
	\end{align}
	with the data-dependent polynomials $\delta_i(\tilde{x}_i^+-Z(\tilde{x}_i,\tilde{u}_i)a)\in\mathbb{R}[a],\,i=1,\dots,D$.
\end{lem}
\begin{proof}
	If $a\in\Sigma_{\TM{\text{SB}}}$ then there exist realizations of the noise $\tilde{d}_i,i=1,\dots,D$ such that $\tilde{x}_i^+=Z(\tilde{x}_i,\tilde{u}_i)a+\tilde{d}_i$ and $\delta_i(\tilde{d}_i)\leq0$. Combining both yields \eqref{DataBasedParaSet}.\\\indent
	To prove the converse, suppose that $a$ is an element of \eqref{DataBasedParaSet}. Then construct $\tilde{d}_i,i=1,\dots,D$ such that $\tilde{x}_i^+=Z(\tilde{x}_i,\tilde{u}_i)a+\tilde{d}_i$. Since $a$ satisfies \eqref{DataBasedParaSet}, $\tilde{d}_i,i=1,\dots,D$ satisfy \eqref{BoundedNoise1}, and hence $a\in\Sigma_{\TM{\text{SB}}}$.
\end{proof}
%Before continuing with the verification of dissipativity properties of the ground-truth system~\eqref{TrueSystem}, we link Definition~\ref{DefFSS1} and Lemma~\ref{LemFSS1} to the feasible system set considered in the set-membership identification literature \cite{Milanese}. 
%\begin{rmk}
%Since the matrices $Q_i,i=1,\dots,D$ can be calculated from the data set \eqref{DataSet}, $\text{FSS}_a$ can be seen as a data-based set-membership model of the unknown system~\eqref{TrueSystem}. A similar set-membership description for nonlinear systems has been examined for set-membership identification \cite{Milanese}. There, Lipschitz bounds on the system dynamics are considered in order to bound the variety of the system dynamics as otherwise there exist infinitely many systems that explain the data. Similarly, the variety of the system dynamics in Definition \ref{DefFSS1} is bounded by the assumption of a polynomial system with bounded degree. 
%\end{rmk}

Since $\text{FSS}_{\TM{\text{SB}}}$ contains the ground-truth systems, \eqref{TrueSystem} is dissipative if all systems of the feasible system set $\text{FSS}_{\TM{\text{SB}}}$ are dissipative. Based on this idea, the following theorem provides a data-based SOS condition for the verification of dissipativity properties without an identified model of \eqref{TrueSystem}.
\begin{thm}\label{ThmDissiVeriGenral}
	Let the data samples \eqref{DataSet} satisfy Assumption~\ref{Noise1}. Then system~\eqref{TrueSystem} is dissipative on \eqref{Constraints} w.r.t. the given supply rate $s\in\mathbb{R}[x,u]$ if there exist a storage function $\lambda\in\text{SOS}[x]$ and polynomials $s_i\in\text{SOS}[x,u,a],i=1,\dots,c$ and $t_i\in\text{SOS}[x,u,a],i=1,\dots,D$ such that $\psi\in\text{SOS}[x,u,a]$ with 
	\begin{align*}
	&\psi(x,u,a)=s(x,u){-}\lambda(Z(x,u)a){+}\lambda(x){+}\dots\notag\\
	& \sum_{i=1}^{D}\delta_i(\tilde{x}_i^+-Z(\tilde{x}_i,\tilde{u}_i)a)t_i(x,u,a){+}\sum_{i=1}^{c}p_i(x,u)s_i(x,u,a).
	\end{align*}
\end{thm}
\begin{proof}
	By Definition~\ref{DissiDef}, all systems of the feasible system set $\text{FSS}_{\TM{\text{SB}}}$, and hence system~\eqref{TrueSystem}, are dissipative on \eqref{Constraints} if there exists a continuous storage function $\lambda:\mathbb{X}\rightarrow\mathbb{R}_{\geq0}$ such that
	\begin{align}\label{Inequ}
	s(x,u)-\lambda(Z(x,u)a)+\lambda(x)\geq0
	\end{align}
	\TM{for all $(x,u)\in\TM{\mathbb{P}}$ and all $a\in\Sigma_{\TM{\text{SB}}}$. Since the sets $\mathbb{P}$ and $\Sigma_{\TM{\text{SB}}}$ are defined by polynomial inequalities in \eqref{Constraints} and Lemma~\ref{LemFSS1}, respectively, we can apply Proposition~\ref{SOSRelaxation} to conclude that \eqref{Inequ}} holds if there exist a storage function $\lambda\in\text{SOS}[x]$ and SOS polynomials $s_i\in\text{SOS}[x,u,a],i=1,\dots,c$ and $t_i\in\text{SOS}[x,u,a],i=1,\dots,D$ such that $\psi\in\text{SOS}[x,u,a]$.	
\end{proof}

Even though $Z(x,u)a$ is an unidentified polynomial vector in $\mathbb{R}[x,u]^{n}$, it is a known polynomial vector in $\mathbb{R}[x,u,a]^{n}$. For that reason, we can verify $\psi\in\text{SOS}[x,u,a]$ as an SOS problem with free variables $x,u$, and $a$ by applying standard SOS solvers, e.g., \cite{YALMIP}. For quadratically bounded noise \eqref{QuadBoundNoise}, we can achieve a\TM{n} SOS condition independent of $a$.

\begin{coro}\label{DissiVeriCoro}
	Let the data samples \eqref{DataSet} satisfy Assumption~\ref{Noise1} with $\delta_i$ from \eqref{QuadBoundNoise}. Then system~\eqref{TrueSystem} is dissipative on \eqref{Constraints} with respect to the supply rate $s\in\mathbb{R}[x,u]$ if there exist a storage function $\lambda(x)=x^TPx, P\succeq0$ and polynomials $t_i\in\text{SOS}[x,u],i=1,\dots,D$ and $s_i(x,u,a)=\begin{bmatrix}a^T&1\end{bmatrix}S_i(x,u)\begin{bmatrix}a^T& 1\end{bmatrix}^T, i=1,\dots,c,$ with $S_i\in\text{SOS}[x,u]^{(n\TM{\ell}+1)\times (n\TM{\ell}+1)}$ such that $\Psi\in\text{SOS}[x,u]^{(n\TM{\ell}+1)\times (n\TM{\ell}+1)}$ with
	\begin{align*}
	\Psi(x,u)=&\sum_{i=1}^{c}p_i(x,u)S_i(x,u)+\sum_{i=1}^{D}Q_it_i(x,u)\\
	\TM{+}& \begin{bmatrix}-Z(x,u)^TPZ(x,u) & 0\\ 0 & s(x,u)+x^TPx
	\end{bmatrix}
	\end{align*}
	and the data-dependent matrices 
	\begin{align*}
	Q_i {=}\begin{bmatrix}\tilde{Z}_i^T\Delta_1\tilde{Z}_i & -\tilde{Z}_i^T(\Delta_1\tilde{x}_i^++\Delta_2)\\
	-({\tilde{x}_i^{+^T}}\Delta_1+\Delta_2^T)\tilde{Z}_i & \begin{bmatrix}{\tilde{x}_i^{+}}\\1	
	\end{bmatrix}^T\begin{bmatrix}\Delta_1 & \Delta_2 \\ \Delta_2^T & \Delta_3\end{bmatrix} \begin{bmatrix}{\tilde{x}_i^{+}}\\1	
	\end{bmatrix}\vspace{0.1cm} \end{bmatrix}
	\end{align*}
	using the abbreviation $Z(\tilde{x}_i,\tilde{u}_i)=\tilde{Z}_i$.
\end{coro}
\begin{proof}
	Note that the quadratically bounded noise \eqref{QuadBoundNoise} yields $\delta_i(\tilde{x}_i^+-Z(\tilde{x}_i,\tilde{u}_i)a) = \begin{bmatrix}a^T&1\end{bmatrix}Q_i\begin{bmatrix}a^T&1\end{bmatrix}^T$. Then pursuing the proof of Theorem~\ref{ThmDissiVeriGenral}, system~\eqref{TrueSystem} is dissipative if there exist a $P\succeq0$, $t_i\in\text{SOS}[x,u],i=1,\dots,D$, and $S_i\in\text{SOS}[x,u]^{(n\TM{\ell}+1)\times (n\TM{\ell}+1)}, i=1,\dots,c$ such that 
	\begin{align}\label{poly1}
	\TM{\begin{bmatrix}a\\1\end{bmatrix}^T\Psi(x,u)\begin{bmatrix}a\\1\end{bmatrix}\in\text{SOS}[x,u,a].}
	\end{align}
	\TM{If $\Psi\in\text{SOS}[x,u]^{(n\TM{\ell}+1)\times (n\TM{\ell}+1)}$ then there exists a $\Phi\in\mathbb{R}[x,u]^{q\times (n\TM{\ell}+1)}$ with $\Psi=\Phi^T\Phi$ by Definition~\ref{DefSOSMatrix}. Therefore, \eqref{poly1} is an SOS polynomial by Definition~\ref{DefSOSMatrix}.}
\end{proof}

For a closer look on Theorem~\ref{ThmDissiVeriGenral} and Corollary~\ref{DissiVeriCoro}, we refer to Section~\ref{SecCompare} and finish this section with an extension of Theorem~\ref{ThmDissiVeriGenral} and Corollary~\ref{DissiVeriCoro}, respectively. %First, we study the case if prior knowledge on the system dynamics~\eqref{TrueSystem} is available and second, parametrized storage function are employed to ease the conservatism of Theorem~\ref{ThmDissiVeriGenral} and Corollary~\ref{DissiVeriCoro} at the small expense of increased computational complexity.
%\begin{rmk}\label{RmkModelKnow}
%If we have prior insight to the system dynamics~\eqref{TrueSystem}, then we could consider instead of \eqref{Systemdescribtion} the system dynamics $x(t+1)=F(x(t),u(t))a+G(x(t),u(t))$, with unidentified coefficients $a\in\mathbb{R}^{n\TM{\ell}}$, known matrix $F\in\mathbb{R}[x,u]^{n\times n\TM{\ell}}$, and known vector $G\in\mathbb{R}[x,u]^n$. %For this system dynamics, we can analogously derive $\text{FFS}_a$ and conditions for verifying dissipativity properties similar to Theorem~\ref{ThmDissiVeriGenral} and Corollary~\ref{DissiVeriCoro}
%\end{rmk}

\begin{rmk}\label{RmkParaStorage}
	To exclude time-varying coefficients $a(t)\in\Sigma_{\TM{\text{SB}}}$ in Theorem~\ref{ThmDissiVeriGenral} and Corollary~\ref{DissiVeriCoro} and hence to reduce their conservatism, we could consider parametrized storage functions $\lambda\in\text{SOS}[x,a]$ %and 
	%\begin{equation*}
	%	\lambda(x,a)=\begin{bmatrix}x\\a\\1\end{bmatrix}^TP
	%	\begin{bmatrix}x\\a\\ 1\end{bmatrix},P\succeq0,
	%\end{equation*}
	%respectively, 
	and the dissipativity inequality
	\begin{align*}
	\lambda(Z(x,u)a,a)-\lambda(x,a)\leq s(x,u),\ \forall (x,u)\in\TM{\mathbb{P}}, \forall a\in\Sigma_{\TM{\text{SB}}}.
	\end{align*}
	
	%Theorem~\ref{ThmDissiVeriGenral} is based on the verification of dissipativity inequality~\eqref{dissipativityInqu} for all systems in the feasible system set $\text{FSS}_a$ using state-depended storage functions. Hence, we actually imply by Theorem~\ref{ThmDissiVeriGenral} dissipativity properties for system~\eqref{Systemdescribtion} with time-varying coefficients $a(t)\in\Sigma_a$. Since the ground-truth system dynamics~\eqref{TrueSystem} are supposed to be polynomial with constant coefficients, state-depended storage functions increase the conservatism of Theorem~\ref{ThmDissiVeriGenral}. To prevent time-varying coefficients and hence improve the accuracy of Theorem~\ref{ThmDissiVeriGenral} and Corollary~\ref{DissiVeriCoro}, respectively, we could consider parametric storage functions $\lambda\in\text{SOS}[x,a]$ or 
	%\begin{equation*}
	%	\lambda(x,a)=\begin{bmatrix}x\\a\\1\end{bmatrix}^TP
	%	\begin{bmatrix}x\\a\\ 1\end{bmatrix},P=\begin{bmatrix}P_{11} & P_{12} & P_{13}\\ P_{12}^T & 0 & 0\\P_{13}^T & 0 & 0\end{bmatrix}\succeq0,
	%\end{equation*}
	%respectively, and the dissipativity inequality
	%\begin{align*}
	%	\lambda(Z(x,u)a,a)-\lambda(x,a)\leq s(x,u),\ \forall (x,u)\in\bb{Z}, \forall a\in\Sigma_a.
	%\end{align*}
	%Since this dissipativity inequality corresponds to the dissipativity of
	%\begin{align*}
	%	x(t+1)&=Z(x(t),u(t))a(t)\\
	%	a(t+1)&=a(t),
	%\end{align*}
	%the coefficients $a\in\Sigma_a$ are constant over time.	
\end{rmk}

\section{Data-driven dissipativity verification for cumulatively bounded noise}\label{SecNonCons}

We again tackle the problem of verifying whether the unidentified polynomial system \eqref{TrueSystem} is dissipative by means of noisy data. However, instead of bounding the noise separately in time as in the previous section, the noise is characterized by one property that bounds cumulatively the noise realizations of the data samples \eqref{DataSet}, which was first proposed in \cite{Groningen}. 

\begin{assu}[Cumulatively bounded noise]\label{Noise2}
	For the measured data \eqref{DataSet}, suppose that the matrix \TM{$\tilde{D}=\begin{bmatrix}\tilde{d}_1&\cdots&\tilde{d}_D\end{bmatrix}$ is an element of} 
	\begin{equation}\label{noise}
	\TM{\mathcal{D}^{\text{CB}}{=}\left\{F\in\mathbb{R}^{n\times D}{:}\begin{bmatrix}F^T\\I_n\end{bmatrix}^T\begin{bmatrix}\varDelta_1 & \varDelta_2\\ \varDelta_2^T & \varDelta_3\end{bmatrix}\begin{bmatrix}F^T\\I_n\end{bmatrix}{\prec0} \right\}}
	\end{equation}
	with $\Delta_1\succeq0$.
\end{assu}

By Assumption~\ref{Noise2}, all noise realizations $\tilde{d}_1,\dots,\tilde{d}_D$ are cumulatively bounded as $\Delta_1\succeq0$. Exemplary, \eqref{noise} incorporates noise with (strictly) bounded energy $\sum_{i=1}^{D}\tilde{d}_i^T\tilde{d}_i<\delta_{\text{e}}^2$ by $\tilde{D}\tilde{D}^T\prec \delta^2_{\text{e}}I_n$. %Furthermore, $\tilde{D}=0$ satisfies Assumption~\ref{Noise2} as $\varDelta_3\prec0$.
%Note that the strictness of \eqref{noise} could be switched with $\Delta_1\succeq0$ while the results of the remainder of this section could be adapted. 
%For more details on the noise description in Assumption~\ref{Noise2}, we refer to \cite{Groningen} and to Section~\ref{SecCompare} for a comparison with Assumption~\ref{Noise1}.
\\\indent
Analogously to \cite{Groningen} and Section~\ref{SecSOS}, combining Assumption~\ref{Noise2}, data samples \eqref{DataSet}, and the system dynamics 
\begin{align}\label{Systemdescribtion2}
x(t+1)=Az(x(t),u(t)),
\end{align}
with unidentified coefficients $A\in\mathbb{R}^{n\times \TM{\ell}}$, yields a data-based set-membership representation of the ground-truth system~\eqref{TrueSystem} which is summarized in the following definition and lemma.
\begin{defn}[Feasible system set]\label{DefFSS2}
	The set of all systems \eqref{Systemdescribtion2} admissible with the measured data set \eqref{DataSet} for cumulatively bounded noise \eqref{noise} is given by the feasible system set $\text{FSS}_{\TM{\text{CB}}}=\{Az\in\mathbb{R}[x,u]^n:A\in\Sigma_{\TM{\text{CB}}}\}$ with $\Sigma_{\TM{\text{CB}}} = \{A\in\mathbb{R}^{n\times\ell}:\TM{\exists \begin{bmatrix}\tilde{d}_1&\cdots&\tilde{d}_D\end{bmatrix}\in\mathcal{D}^{\text{CB}}} \text{\ satisfying\ }\tilde{x}_i^+=Az(\tilde{x}_i,\tilde{u}_i)+\tilde{d}_i,i=1,\dots,D\}$.
\end{defn}

\begin{lem}\label{LemSigmaA}
	The set of all coefficients $\Sigma_{\TM{\text{CB}}}$ for which system \eqref{Systemdescribtion2} explains the measured data set \eqref{DataSet} for cumulatively bounded noise~\eqref{noise} is equivalent to
	\begin{align}\label{FSSA1}
	\left\{A\in\mathbb{R}^{n\times\ell}:\begin{bmatrix}A^T\\I_n\end{bmatrix}^T\begin{bmatrix}\tilde{\varDelta}_1 & \tilde{\varDelta}_2\\ \tilde{\varDelta}_2^T & \tilde{\varDelta}_3\end{bmatrix}\begin{bmatrix}A^T\\I_n\end{bmatrix}\prec0\right\}
	\end{align}
	with the data-dependent matrices $\tilde{X}^+=\begin{bmatrix}\tilde{x}_1^+&\cdots&\tilde{x}_D^+\end{bmatrix}$, $\tilde{Z}=\begin{bmatrix}z(\tilde{x}_1,\tilde{u}_1)&\cdots&z(\tilde{x}_D,\tilde{u}_D)\end{bmatrix}$, and
	\begin{align*}
	&\begin{bmatrix}\tilde{\varDelta}_1 & \tilde{\varDelta}_2\\ \tilde{\varDelta}_2^T & \tilde{\varDelta}_3\end{bmatrix}\\
	&\hspace{0.2cm} {=}\begin{bmatrix}\tilde{Z}\varDelta_1\tilde{Z}^T & -\tilde{Z}(\varDelta_1\tilde{X}^{+^T}+\varDelta_2)\\-(\tilde{X}^+\varDelta_1{+}\varDelta_2^T)\tilde{Z}^T & \begin{bmatrix}\tilde{X}^{+^T}\\I_n\end{bmatrix}^T\begin{bmatrix}\varDelta_1 & \varDelta_2\\ \varDelta_2^T & \varDelta_3\end{bmatrix}\begin{bmatrix}\tilde{X}^{+^T}\\I_n\end{bmatrix} \end{bmatrix}.
	\end{align*}
\end{lem}\vspace{0.2cm}
\begin{proof}
	The statement follows analogously to \cite{Groningen} (Lemma 4) and the proof of Lemma~\ref{LemFSS1}, respectively.
\end{proof}

\TM{Since the data-based description of $\Sigma_{\TM{\text{CB}}}$ in Lemma~\ref{LemSigmaA} provides a bound on $A^T$ instead of $A$ as will be required for the verification of the \textquotedblleft primal" dissipativity inequality~\eqref{dissipativityInqu}}, we introduce the dual version of \eqref{FSSA1} as in \cite{AnneDissi}.

\begin{lem}\label{LemmaDual}
	Suppose that Assumption~\ref{Noise2} holds and the inverse
	\begin{equation}\label{InverseCond}
	\begin{bmatrix}-\tilde{\varDelta}_1 & \tilde{\varDelta}_2\\ \tilde{\varDelta}_2^T & -\tilde{\varDelta}_3\end{bmatrix}^{-1}=:\begin{bmatrix}\bar{\varDelta}_1 & \bar{\varDelta}_2\\ \bar{\varDelta}_2^T & \bar{\varDelta}_3\end{bmatrix}
	\end{equation}
	exists. Then any matrix $A\in\mathbb{R}^{n\times \TM{\ell}}$ is an element of $\Sigma_{\TM{\text{CB}}}$ if and only if 
	\begin{equation*}\label{FSSA2}
	A\in\overline{\Sigma}_{\TM{\text{CB}}}=\left\{A\in\mathbb{R}^{n\times\ell}:\begin{bmatrix}I_{\TM{\ell}}\\A\end{bmatrix}^T\begin{bmatrix}\bar{\varDelta}_1 & \bar{\varDelta}_2\\ \bar{\varDelta}_2^T & \bar{\varDelta}_3\end{bmatrix}\begin{bmatrix}I_{\TM{\ell}}\\A\end{bmatrix}\prec0\right\}.
	\end{equation*}
\end{lem}
\begin{proof}
	\TM{Since the samples \eqref{DataSet} satisfy $\tilde{x}_i^+=f(\tilde{x}_i,\tilde{u}_i)+\tilde{d}_i$ with $\begin{bmatrix}\tilde{d}_1&\cdots&\tilde{d}_D\end{bmatrix}\in\mathcal{D}^{\text{CB}}$ by assumption, the coefficient matrix $A_{\text{gt}}$ of the ground-truth system~\eqref{TrueSystem}, i.e., $f(x,u)=A_{\text{gt}}z(x,u)$, is an element of $\Sigma_{\TM{\text{CB}}}$. Together with $\Delta_1\succeq0$, the dualization lemma~\cite{SchererLMI} implies that ($A_{\text{gt}}\in\overline{\Sigma}_{\TM{\text{CB}}}$ and) $\bar{\varDelta}_3\succeq0$. Thereby, any matrix $A\in\mathbb{R}^{n\times \TM{\ell}}$ satisfies $A\in\Sigma_{\TM{\text{CB}}}$ if and only if $A\in\overline{\Sigma}_{\TM{\text{CB}}}$ again by the dualization lemma.}
\end{proof}

\TM{By Lemma~\ref{LemmaDual}}, the feasible system sets $\text{FSS}_{\TM{\text{CB}}}$ and $\overline{\text{FSS}}_{\TM{\text{CB}}}=\{Az\in\mathbb{R}[x,u]^n:A\in\bar{\Sigma}_{\TM{\text{CB}}}\}$ are equivalent and contain the ground-truth system~\eqref{TrueSystem}. Therefore, we can derive analogously to Section~\ref{SecSOS} a condition to verify dissipativity properties of polynomial system~\eqref{TrueSystem} without identifying a model but directly from noisy input-state measurement. 
\begin{thm}\label{ThmDissiVeriNonCons}
	Suppose that the data samples \eqref{DataSet} satisfy Assumption~\ref{Noise2}, the inverse \eqref{InverseCond} exists, the state-inputs constraints \eqref{Constraints} are specified by
	\begin{equation}\label{Constraints1}
	p_i(x,u)= \begin{bmatrix}z(x,u)\\1\end{bmatrix}^TP_{i}\begin{bmatrix}z(x,u)\\1\end{bmatrix},i=1,\dots,c,
	\end{equation}	
	with $P_i\in\mathbb{R}^{(\ell+1)\times(\ell+1)}$ and, without loss of generality, there exist matrices $T_x\in\mathbb{R}^{n\times \TM{\ell}}$ and $T\in\mathbb{R}^{(n+m)\times \TM{\ell}}$ such that $x=T_x z$ and $\begin{bmatrix}x\\u\end{bmatrix}=Tz$. Then system~\eqref{TrueSystem} is $(Q,S,R)$-dissipative on \eqref{Constraints} with quadratic constraints \eqref{Constraints1} if \TM{the LMI~\eqref{LMICond} 
		\begin{figure*}
			\begin{align}\label{LMICond}
			\Theta:=\begin{bmatrix}\begin{array}{c}	\begin{matrix}I_n & 0 & 0\\ 0 & T_x & 0\end{matrix}\\\hline\begin{matrix}0 & T & 0\end{matrix}\\\hline \begin{matrix}0 &I_{\TM{\ell}} & 0\\ I_n& 0 & 0\end{matrix}\\\hline \begin{matrix}0 & I_{\TM{\ell}} & 0\\ 0 & 0 & 1\end{matrix}\end{array}\end{bmatrix}^T
			\begin{bmatrix}\begin{array}{c|c|c|c}
			\begin{matrix}\TM{-}P & 0 \\ 0 & P\end{matrix} & \begin{matrix}0 & 0 \\ 0 & 0 \end{matrix} & \begin{matrix}0 & 0 \\ 0 & 0\end{matrix} & \begin{matrix} 0 \\ 0\end{matrix}\\\hline
			\begin{matrix}0 & 0 \\ 0 & 0\end{matrix} & \begin{matrix}Q & S \\ S^T & R \end{matrix} & \begin{matrix}0 & 0 \\ 0 & 0\end{matrix} & \begin{matrix} 0 \\ 0\end{matrix}\\\hline
			\begin{matrix}0 & 0 \\ 0 & 0\end{matrix} & \begin{matrix}0 & 0 \\ 0 & 0 \end{matrix} & \begin{matrix}\tau\bar{\varDelta}_1\phantom{\Big|} & \tau\bar{\varDelta}_2 \\ \tau\bar{\varDelta}_2^T & \tau\bar{\varDelta}_3\end{matrix} & \begin{matrix} 0 \\ 0\end{matrix}\\\hline		
			\begin{matrix}0 & 0 \end{matrix} & \begin{matrix}0 & 0 \end{matrix} & \begin{matrix}0 & 0\end{matrix} & \sum_{i=1}^{c} \tilde{P}_i(\tau_i)	\phantom{\Big|}
			\end{array}		
			\end{bmatrix}	
			\begin{bmatrix}\begin{array}{c}	\begin{matrix}I_n & 0 & 0\\ 0 & T_x & 0\end{matrix}\\\hline\begin{matrix}0 & T & 0\end{matrix}\\\hline \begin{matrix}0 &I_{\TM{\ell}} & 0\\ I_n& 0 & 0\end{matrix}\\\hline \begin{matrix}0 & I_{\TM{\ell}} & 0\\ 0 & 0 & 1\end{matrix}\end{array}\end{bmatrix}\TM{\succeq}\, 0
			\end{align}
			\begin{tikzpicture}
			\draw[-,line width=0.7pt] (0,0) -- (18,0);
			\end{tikzpicture}
		\end{figure*}
		holds for a storage function $\lambda(x)=x^TPx, P\succeq0$, a constant $\tau\geq0$, and polynomials $z_i\tau_i\in\text{SOS}[x,u],i=1,\dots,c$ with a vector of monomials $z_i\in\mathbb{R}[x,u]^{1\times\beta}$, to-be-optimized coefficients $\tau_i\in\mathbb{R}^\beta$, and a linear mapping $\tilde{P}_i:\mathbb{R}^\beta\rightarrow\mathbb{R}^{(\ell+1)\times(\ell+1)}$ with
		\begin{align}\label{QuadraticDecomp}
		z_i\tau_i\begin{bmatrix}z\\1\end{bmatrix}^TP_i\begin{bmatrix}z\\1\end{bmatrix} = \begin{bmatrix}z\\1\end{bmatrix}^T\tilde{P}_i(\tau_i)\begin{bmatrix}z\\1\end{bmatrix}.
		\end{align}}
\end{thm}
\begin{proof}
	Since the ground-truth system~\eqref{TrueSystem} is an element of $\text{FSS}_{\TM{\text{CB}}}$ and $\overline{\text{FSS}}_{\TM{\text{CB}}}$ by Lemma~\ref{LemmaDual}, system \eqref{TrueSystem} is $(Q,S,R)$-dissipative on \eqref{Constraints} with quadratic constraints \eqref{Constraints1} if there exists a storage function $\lambda(x)=x^TPx, P\succeq0$ such that
	\begin{equation}\label{CondDissiIn}
	\begin{aligned}
	&x^TPx-z(x,u)^TA^TPAz(x,u)+\begin{bmatrix}x\\u\end{bmatrix}^T\begin{bmatrix}Q & S\\ S^T & R\end{bmatrix}\begin{bmatrix}x\\u\end{bmatrix}\TM{\geq}0,\\ 
	&\forall (x,u):\begin{bmatrix}z(x,u)\\1\end{bmatrix}^TP_i\begin{bmatrix}z(x,u)\\1\end{bmatrix}{\leq} 0, i=1,\dots,c,\forall A{\in}\overline{\Sigma}_{\TM{\text{CB}}}.
	\end{aligned}
	\end{equation}
	\TM{With $A{\in}\overline{\Sigma}_{\TM{\text{CB}}}$ implying that for all $(x,u)\in\mathbb{R}^n\times\mathbb{R}^m$
		\begin{equation*}
		z(x,u)^T\begin{bmatrix}I_{\TM{\ell}}\\A\end{bmatrix}^T\begin{bmatrix}\bar{\varDelta}_1 & \bar{\varDelta}_2\\ \bar{\varDelta}_2^T & \bar{\varDelta}_3\end{bmatrix}\begin{bmatrix}I_{\TM{\ell}}\\A\end{bmatrix}z(x,u)\leq0,
		\end{equation*}
		we apply Proposition~\ref{SOSRelaxation} to conclude that the conditioned dissipativity inequality \eqref{CondDissiIn} holds if there exist a $P\succeq0$, a non-negative constants $\tau$ and polynomials $z_i\tau_i\in\text{SOS}[x,u],i=1,\dots,c$ with \eqref{QuadraticDecomp} satisfying 
		\begin{equation}\label{Ineq2}
		L(x,u,A)^T\Theta L(x,u,A)\in\text{SOS}[x,u,\text{vec}(A)]
		\end{equation}
		with $L(x,u,A)=\begin{bmatrix}	z(x,u)^TA^T & z(x,u)^T &  1\end{bmatrix}^T$.
		%\begin{align*}
		%	*^T
		%	\Theta	
		%	\begin{bmatrix}\begin{array}{c}	Az(x,u)\\T_xz(x,u)\\\hline Tz(x,u)\\\hline z(x,u)\\Az(x,u)\\\hline z(x,u)\\1\end{array}\end{bmatrix}\leq0
		%\end{align*}
		%To attain a tractable LMI condition, we extract the matrix $A$ and the nonlinear proportion $z(x,u)$
		Finally, if \eqref{LMICond} is satisfied then there exists a matrix $\Omega$ with $\Theta=\Omega^T\Omega$ and thus \eqref{Ineq2} is an SOS polynomial by Definition~\ref{DefSOSMatrix}.} %Moreover, since $\tilde{P}_i$ of the quadratic decomposition \eqref{QuadraticDecomp} contains linearly the to-be-optimized coefficients of the SOS polynomial $\tau_i(x,u)$, \eqref{LMICond} is indeed an LMI.
\end{proof}

In Theorem~\ref{ThmDissiVeriNonCons}, dissipativity verification boils down to an LMI feasibility problem instead of a\TM{n} SOS problem as in Theorem~\ref{ThmDissiVeriGenral} because \TM{we extract in \eqref{Ineq2} all monomials in $x,u$, and $\text{vec}(A)$ into $L(x,u,A)$ similar to the square matricial representation in Proposition~\ref{SOSMatrix}.} %the implication of \eqref{Ineq2} by LMI \eqref{LMICond} corresponds to a\TM{n} SOS relaxation. 
\TM{However, this computational advantage comes at the cost of additional conservatism compared to Theorem~\ref{ThmDissiVeriGenral} as Theorem~\ref{ThmDissiVeriNonCons} considers quadratic storage functions and requires in its derivation simplified multipliers for Proposition~\ref{SOSRelaxation}, e.g., $\tau$ independent of $x$ and $u$.}  Note that we can generalize Theorem~\ref{ThmDissiVeriNonCons} for supply rates $s(x,u)=z(x,u)^TQz(x,u)$.%Moreover, note that we proceed similarly as for providing quadratic performance guarantees for linear fractional representations which are exploited in \cite{AnneDissi} to verify dissipativity properties for unknown linear systems.
\\\indent
In Theorem~\ref{ThmDissiVeriNonCons}, we consider SOS polynomials \TM{$z_i(x,u)\tau_i,i=1,\dots,c$} instead of non-negative constants as otherwise LMI \TM{\eqref{LMICond}} becomes indefinite if $P_i$ contains a negative right lower element which is mostly the case, e.g., $x^Tx\leq1$. Note that \TM{a linear mapping $\tilde{P}_i$ exists as the left hand side of the quadratic decomposition \eqref{QuadraticDecomp} is linear in $\tau_i$}. However, $\tilde{P}_i$ is not unique but is spanned by a linear subspace which provides additional degrees of freedom to deteriorate the conservatism of condition \eqref{LMICond}. \\\indent
We conclude this section by demonstrating the flexibility of this framework by employing prior system knowledge.% and appending additional nonlinearities and uncertainties.

\begin{rmk}\label{RmkExtensions}
	We can take prior knowledge of the system dynamics into account by considering 
	\begin{equation*}
	x(t+1)=Az_1(x(t),u(t))+\begin{bmatrix}\bar{A}_1 & \bar{A}_2\end{bmatrix}\begin{bmatrix}z_1(x(t),u(t))\\z_2(x(t),u(t))\end{bmatrix}
	\end{equation*}
	with unidentified matrix $A$ and known matrices $\bar{A}_1$ and $\bar{A}_2$. The additional vector of monomials $z_2(x,u)$ is beneficial if, for instance, $g(x)$ of a (polynomial) control-affine system $x(t+1)=f(x(t))+g(x(t))u(t)$ is known from some insight to the system. Moreover, $z_2(x,u)$ might be necessary for the quadratic decomposition \eqref{QuadraticDecomp}. \TM{Note that incorporating prior knowledge as described here is also conceivable for the framework of separately bounded noise.}%\\\indent
	%Inspired by the extraction of the nonlinearity $z(x,u)$ in inequality~\eqref{Ineq2}, a third extension might be the consideration of quadratically bounded non-polynomial nonlinearities $g(x,u)$
	%\begin{equation*}
	%	x(t+1)=A\begin{bmatrix}z(x(t),u(t))\\g(x(t),u(t))\end{bmatrix}
	%\end{equation*}	
	%where $z(x,u)$ still contains only monomials in $x$ and $u$. %Potentially, the additional nonlinearities could be bounded dynamically using integral quadratic constraints in discrete time. However, note that non-polynomial nonlinearities lead to the loss of some properties of the SOS relaxation as the asymptotically exactness mentioned in Section~\ref{SecSOS}. 
	%Note that this extension is also conceivable for the first framework for separately bounded noise.
\end{rmk}

\section{Comparison of both frameworks for separately bounded noise}\label{SecCompare}

Motivated by the frequently assumed separately bounded noise $||\tilde{d}_i||_2\leq \epsilon_i$ as non-probabilistic noise description, e.g., in system identification \cite{Milanese}, we compare in this section both previously proposed frameworks for data-driven dissipativity verification for this noise characterization. \\\indent
According to \cite{Groningen}, the cumulatively bounded noise description \eqref{noise} can incorporated this separately bounded noise by 
$\tilde{D}\tilde{D}^T\preceq \sum_{i=1}^{D}\epsilon_i^2I_n$. However, this characterisation also facilitates, e.g., noise with bounded energy $\sum_{i=1}^{D}\tilde{d}_i^T\tilde{d}_i\leq\sum_{i=1}^{D}\epsilon_i^2$ which includes more noise realizations than $||\tilde{d}_i||_2\leq \epsilon_i$. Hence, Assumption~\ref{Noise1} provides a more accurate description than Assumption~\ref{Noise2} \TM{for the separately bounded noise $||\tilde{d}_i||_2\leq \epsilon_i$}, and therefore leads to a tighter set-membership representation of the ground-truth system \eqref{TrueSystem}. For that reason, Theorem~\ref{ThmDissiVeriGenral} provides a less conservative condition for dissipativity verification than Theorem~\ref{ThmDissiVeriNonCons} which is indeed observed in Section~\ref{SecEx}.\\\indent
Furthermore, the feasible system set $\text{FSS}_{\TM{\text{SB}}}$ cannot increase by considering additional data samples. Contrary, we show in Subsection~\ref{SecEx1} that adding samples with high signal-to-noise-ratio \TM{to an original data set of $\text{FSS}_{\TM{\text{CB}}}$ might decrease its accuracy for dissipativity verification, and hence might render} LMI~\eqref{LMICond} infeasible. \TM{One explanation is that we cumulate all data samples equally weighted in \eqref{FSSA1} into one condition which corresponds to restrict all SOS polynomial multipliers $t_i(x,u,a),i=1,\dots,D$ in Theorem~\ref{ThmDissiVeriGenral} to be equal. Therefore, data samples with large noise increase the uncertainty of $\text{FSS}_{\TM{\text{SB}}}$. To circumvent this problem in Theorem~\ref{ThmDissiVeriNonCons}, we could consider  the intersection of $\text{FSS}_{\TM{\text{CB}}}$ for the original data set and for the data set with additional data by an S-procedure argument.}
%To some extend, this observation is important, e.g., if the data set includes outliers because for the cumulative noise description they influence negatively the whole data set if we don't neglect them.  
\\\indent 
Further advantages of Theorem~\ref{ThmDissiVeriGenral} are that its accuracy can be improved by parametrized storage functions as shown in Remark~\ref{RmkParaStorage} and general polynomial state-input constraints and supply rates can be handled.\\\indent
On the other hand, the framework of cumulatively bounded noise is computationally more attractive. The verification condition in Theorem~\ref{ThmDissiVeriNonCons} boils down to an LMI condition and its complexity doesn't increase with the amount of samples as all data samples \eqref{DataSet} are cumulated into one condition. Contrary, Theorem~\ref{ThmDissiVeriGenral} requires one additional SOS polynomial multiplier for each sample which might yield to a non-tractable optimization problem. This issue could be circumvented by the relaxation $t_1(x,u,a)=\dots=t_D(x,u,a)$ which then leads to a cumulative noise characterization. \\\indent
Furthermore, in our testing in Section~\ref{SecEx}, system description~\eqref{Systemdescribtion2} is computationally more efficient than \eqref{Systemdescribtion} when tackling systems \eqref{Systemdesc} with a large number of unidentified coefficients.\\\indent
To summarize this discussion, while the framework of separately bounded noise provides a data-efficient approach for the often used bounded noise $||\tilde{d}_i||_2\leq \epsilon_i$, the framework of cumulatively bounded noise is computationally more attractive. For that reason, the latter framework should always be considered if the noise is characterized by some cumulative property.

\section{Numerical Examples}\label{SecEx}

To measure the conservatism of both frameworks for separately bounded noise, we apply Corollary~\ref{DissiVeriCoro} and Theorem~\ref{ThmDissiVeriNonCons} on two systems to find a guaranteed upper bound on their $\ell_2$-gain $\gamma$ which corresponds to the supply rate $s(x,u)=\gamma^2 u^Tu-x^Tx$. To this end, the SOS problem of Corollary~\ref{DissiVeriCoro} and the LMI feasibility problem of Theorem~\ref{ThmDissiVeriNonCons} are extended by the minimization over $\gamma$.% and are solved in Matlab using YALMIP \cite{YALMIP} and the solver MOSEK. 

\subsection{Example 1}\label{SecEx1}

We determine an upper bound on the $\ell_2$-gain of the polynomial system
\begin{align*}
x(t+1)=-0.8x(t)+0.1x(t)^2+u(t)
\end{align*}
with state constraint $x^2-1\leq 0$ and input constraint $u^2-0.01\leq 0$. We receive the upper bound $\gamma\leq10.0$ by SOS optimization exploiting the system dynamics.\\\indent
To apply our data-driven methods, we draw samples \TM{\eqref{DataSet}} from a single trajectory with initial condition $x(0)=1$, input $u(t)=0.1,t\geq0$, and a random sampled and (separately) bounded noise $|\tilde{d}_i|\leq0.02$.% For Corollary~\ref{DissiVeriCoro}, we use a parametrized storage function, quadratic SOS polynomials $s_i(x,u),i=1,2$, and quartic SOS polynomials $t_i(x,u),i=1,\dots,D$. In Theorem~\ref{ThmDissiVeriNonCons}, quadratic SOS polynomials $\tau_i(x,u),i=1,2$ are considered. All optimization problems are solved in less than a second on a Lenovo i5 notebook.
\\\indent
Considering the first three \TM{noisy} data samples of the trajectory, we receive the upper bounds for the $\ell_2$-gain $\gamma_{\text{SB}}=16.3$ from Corollary~\ref{DissiVeriCoro} and $\gamma_{\text{CB}}=17.1$ from Theorem~\ref{ThmDissiVeriNonCons}. \\\indent
As stated in Section~\ref{SecCompare}, additional data don't increase $\gamma_{\TM{\text{SB}}}$ but potentially $\gamma_{\TM{\text{CB}}}$. Indeed, while the upper bound $\gamma_{\text{SB}}$ decreases to $13.3$ using the first $20$ samples, $\gamma_{\text{CB}}$ increases to $74.7$ using the first $6$ samples and LMI \eqref{LMICond} even becomes infeasible for more samples. This observation is due to the high signal-to-noise-ratio of the measured trajectory for $t\geq 5$. Note that all optimization problems in this example are solved in less than a second on a Lenovo i5 notebook.

\subsection{Example 2}\label{SecEx2}

The $\ell_2$-gain of the system
\begin{align*}
\begin{bmatrix}x_1(t+1)\\x_2(t+1)
\end{bmatrix}=\begin{bmatrix}
-0.5x_1+0.3x_2^2+0.2x_1x_2\\
0.4x_2+0.1x_2^2-0.2x_1^3+u
\end{bmatrix}(t)
\end{align*}
with $x_1^2\leq 1$, $x_2^2\leq 1$, and $u^2\leq 1$ is examined. Given the ground-truth system, we determine $2.1$ as an upper bound of the $\ell_2$-gain \TM{by SOS optimization}. The noise of the data \TM{\eqref{DataSet}} exhibits constant signal-to-noise-ratio $||\tilde{d}_i||_2\leq0.02||\tilde{x}_i||_2$. Furthermore, $x(0)=\begin{bmatrix}-1 & -1\end{bmatrix}^T$ and we apply the input signal $u(t)=0.7\sin(0.002t^2+0.1t)$ such that the system is excited over the whole time horizon. Moreover, we assume \TM{$z(x,u)=\begin{bmatrix}x_1&x_2& x_1^2&x_2^2&x_1x_2&x_1^3&u\end{bmatrix}^T$, i.e., the unidentified \TM{model} \eqref{Systemdesc} contains $14$ unknown coefficients and more monomials than is required to describe the ground-truth system}.\\\indent
Using the first \TM{$30$ noisy} samples of the input-state trajectory, we calculate the bounds $\gamma_{\text{SB}}=\TM{3.8}$ and $\gamma_{\text{CB}}=\TM{11.1}$. With \TM{$300$} data samples available, we can reduce the upper bound $\gamma_{\text{SB}}$ to $\TM{2.3}$ and $\gamma_{\text{CB}}$ to $\TM{3.6}$. \TM{The advantage of Theorem~\ref{ThmDissiVeriNonCons} is that computation time to solve its optimization problem is about two seconds while solving the SOS optimization problem of Corollary~\ref{DissiVeriCoro} takes now about $10\,\text{minutes}$.}

%While the LMI \eqref{LMICond} is infeasible when increasing the signal-to-noise-ration to $||d||_2\leq0.04||x||_2$, Theorem~\ref{ThmDissiVeriGenral} still provides an upper bound of $\gamma_{\text{SB}}=80.7$ for $D=20$.\\\indent
%Note that both frameworks determine meaningful bounds on the $\mathcal{L}_2$-gain with less data samples than \cite{MontenbruckLipschitz} at the cost of input-state measurements and a polynomial description of the system which requires more insight to the system. For example, 
%\cite{MontenbruckLipschitz} estimates the $\mathcal{L}_2$-gain of a similar complex system by approximately $10^4$ data samples.

\section{Conclusions}

We established two set-membership frameworks to check whether a polynomial system is dissipative without an explicitly \TM{identified} model but directly from noise-corrupted input-state measurements. The first framework provides a data-efficient but computationally expensive condition for separately bounded noise using standard SOS optimization. The second framework considers cumulatively bounded noise to deduce a more computationally attractive LMI condition with SOS multipliers, which corresponds partially to a generalization of \cite{AnneDissi} for polynomial systems. %In numerical examples, we showed that both frameworks are more data-efficient than using Lipschitz approximations \cite{MontenbruckLipschitz}, \cite{IterativeNLM} at the cost of input-state measurements and a polynomial description of the system which requires more insight to the system.
Subject of future research, we extend the results to find optimal dissipativity properties as conic relations~\cite{Zames} or nonlinearity measures \cite{MartinNLM}. Furthermore, the extension of the presented frameworks for input-output measurements might be interesting.

\end{document}